\documentclass{article}
\usepackage[shortlabels]{enumitem}
\usepackage{titling}

\usepackage{palatino} 
\usepackage{mathpazo}
\usepackage{braket}
\usepackage{amsfonts}
\usepackage{amssymb}
\usepackage{amsmath}
\usepackage{mathrsfs} 
\usepackage{bbm}
\usepackage{latexsym}
\usepackage{amsthm}
\usepackage[usenames]{color}
\usepackage{hyperref, etoolbox}
\usepackage{diagbox}
\usepackage{subcaption}
\usepackage{verbatim}
\usepackage{ragged2e, graphicx, pifont}
\usepackage{indentfirst}
\usepackage{algorithm}
\usepackage[noend]{algpseudocode}
\usepackage{tikz}
\usetikzlibrary{decorations.pathmorphing}
\usetikzlibrary{arrows.meta}
\usetikzlibrary{positioning, calc}
\usepackage{quantikz}

\tikzset{
gate/.style={draw,minimum width=6mm,minimum height=9mm},
cwire/.style={double distance=1pt}
}

\usepackage{mdframed}
\usepackage{thm-restate}

\hypersetup{
colorlinks = true,
citecolor= blue,
linkcolor= blue
}

\usepackage{tabularx}

\hypersetup{pdfpagemode=UseNone}

\newtheorem{theorem}{Theorem}
\newtheorem{lemma}[theorem]{Lemma}

\newtheorem{cor}[theorem]{Corollary}
\newtheorem{definition}{Definition}

\newtheorem{fact}{Fact}

\newcommand{\eps}{\varepsilon}

\newcommand{\clA}{\mathcal{A}}
\newcommand{\clB}{\mathcal{B}}

\newcommand{\clP}{\mathcal{P}}

\newcommand{\clX}{\mathcal{X}}
\newcommand{\clY}{\mathcal{Y}}

\newcommand{\sfH}{\mathsf{H}}
\newcommand{\sfI}{\mathsf{I}}

\newcommand{\sfV}{\mathsf{V}}

\newcommand{\bbE}{\mathop{\mathbb{E}}}

\newcommand{\Tr}{\mathrm{Tr}}
\newcommand{\poly}{\mathrm{poly}}

\newcommand{\negl}{\mathrm{negl}}

\newcommand{\chsh}{\mathrm{CHSH}}
\newcommand{\Imax}{\sfI_{\max}}
\newcommand{\Is}{\sfI_{\mathrm{s}}}

\title{Quantum information advantage based on Bell inequalities}
\author{Rahul Jain \\
\small Centre for Quantum Technologies, \\
\small Department of Computer Science, \\
\small National University of Singapore \\
\small \texttt{dcsrahul@nus.edu.sg} \\
\and Srijita Kundu \\
\small Quantum Computing Research Centre, \\
\small Hon Hai (Foxconn) Research Institute \\
\small \texttt{srijita.kundu@foxconn.com.sg} }
\date{}

\begin{document}

\maketitle

\begin{abstract}
    Recently, Kretschmer et al.~\cite{KGD+25} presented an experimental demonstration of a proposed {\em quantum information advantage} protocol. 

    We present an alternate proposal based on a relation derived from parallel-repeated CHSH games. Our memory measure is based on an information measure and is different from~\cite{KGD+25}, where they count the number of qubits. Our proposal has an efficient verifier and a noise-robust quantum prover which is arguably much more efficient compared to~\cite{KGD+25}.  
\end{abstract}

\section{Introduction}
An important line of research in quantum information science is the demonstration of quantum advantage in computational tasks, i.e., the demonstrating that some problems can be solved with fewer quantum resources than classical resources of a comparable kind.

One of the earliest examples of quantum advantage is Bell nonlocality. In a nonlocal game, spatially separated parties Alice and Bob get inputs that are private to them, and are required to produce outputs that satisfy some conditions. The parties may deploy a classical strategy with shared randomness, or a quantum strategy with shared entanglement in order to win such a game. In some nonlocal games, a quantum strategy is able to achieve a much higher winning probability than a classical strategy, in what is known as a Bell inequality violation.

However, the actual computational task behind a Bell inequality violation is not computationally difficult with classical resources --- the difficulty only arises because the parties are separated and cannot communicate (and thus know each other's inputs). Later demonstrations of quantum advantage focused on tasks that are computationally difficult for a classical computer, but easy even with a near-term or currently existing quantum device.

Most of the work in demonstrating quantum advantage in computationally difficult tasks has been sampling-based, showing that quantum devices can efficiently sample from a distribution (essentially the output distribution of a quantum circuit) that a classical device cannot sample efficiently \cite{Aru19, Wu21, Zhu22, Mad22, Mor24, Gao25, Dec25}. However, we cannot actually \emph{prove} that this sampling tasks are hard for classical devices. There is some complexity theoretic evidence \cite{AA13, AG20} that these tasks are hard, but the claimed quantum advantage here ultimately still relies on unproven complexity theoretic assumptions.

A recent work by Kretschmer et al. \cite{KGD+25} aimed to mitigate this situation by demonstrating a sampling task in which quantum advantage can be proven \emph{unconditionally}. The catch here is that the quantum advantage is not in running time, but in \emph{memory}. The authors showed that there is a computational task that a quantum device can solve with few ($12$) qubits of memory that a classical device provably requires a much higher number (between $62$ and $382$) of bits of memory to solve. They called their result \emph{quantum information advantage}.

The construction of \cite{KGD+25} is based on a problem from the setting of communication complexity. In the communication complexity setting, like in the nonlocal game setting, two spatially separated parties Alice and Bob get private inputs, and have to communicate either classically or quantumly in order to compute a relation on these inputs (with the aim of communicating as little as possible). Many separations between quantum and classical communication costs of problems have previously been shown. The work~\cite{KGD+25} translated a communication complexity separation (in a two devices setting) to a separation between the classical and quantum memory cost in a single device setting.

More specifically, \cite{KGD+25} consider a task that separates \emph{one-way} classical and quantum communication complexity, where there is a restriction that there may only be a single classical or quantum message from Alice to Bob. In their setting for quantum information advantage, they recast Alice and Bob not as two parties separated in space, but as a single device separated in \emph{time}. The device gets Alice's input $x$ for the communication problem at time $t_0$, and Bob's input $y$ at a later time $t_1$. The quantity of interest is how much classical or quantum memory the device needs to store between times $t_0$ and $t_1$ in order to compute the relation of interest on $x$ and $y$. It is easy to see that the memory requirement here is precisely equal to the amount Alice needs to communicate about her input $x$ to Bob in the communication complexity setting, thus allowing the use of a classical vs quantum separation in communication complexity.

\subsection{Our proposal}
We propose a different problem that can be used to demonstrate a quantum information advantage, using a different measure of memory than \cite{KGD+25}. Unlike \cite{KGD+25}, we consider as our measure the information about the input $x$ that has to be passed from time $t_0$ to $t_1$, as quantified by the measure {\em smoothed max mutual information}, $\sfI^\delta_{\max}(X:M)$, where $M$ is the memory register  maintained by the prover between $t_0$ and $t_1$.. This is related to the more standard {\em mutual information} between $X$ and $M$, $\sfI(X:M)$; see Section \ref{sec:prelims} for a definition of both these quantities. We think this is a valid measure for quantum information advantage, since it is more difficult to maintain qubits in memory in such a state that they retain information about inputs, than to maintain them in any state whatsoever. We formally define a proof of quantum information advantage as follows.
\begin{definition}[Proof of quantum information advantage]
Let $c(n), s(n)$ be polynomial functions of $n$, and let $\gamma_0$ be a constant in $[0,1]$. A $(c,s,\gamma_0)$ noise-robust proof of quantum information advantage is an interactive protocol between a verifier which is a PPT (probabilistic polynomial time) algorithm and a prover which could be a PPT or a QPT (quantum polynomial time) algorithm. In the protocol, the verifier tests that the prover is able to compute a relation $R(n) \subseteq (\clX(n)\times\clY(n))\times(\clA(n)\times\clB(n))$, where $n$ is the input parameter. 

\begin{enumerate}
\item The verifier and the prover both get as input $1^n$.
    \item At time $t_0$, the verifier samples  $x\in \clX$ and sends to the prover. Let $X$ denote the random variable corresponding to $x$. 
    \item The prover replies with $a\in\clA$. 
    \item Subsequently, at time $t_1$, the verifier samples $y\in\clY$ and sends to the prover. 
    \item The prover replies with $b\in\clB$. 
    \item The verifier accepts if $(x,y,a,b) \in R$. 
\end{enumerate}
We have the following requirements.
\begin{itemize}
\item Completeness: There exists a QPT prover that can make the verifier accept with probability $1-\negl(n)$. The quantum memory register $M$ maintained by the prover between $t_0$ and $t_1$ satisfies $\sfI^{\negl(n)}_{\max}(X:M)\leq c(n)$. 
\item Noise-robustness: There exists a noise threshold $\gamma_0 \in (0,1)$ such that if the QPT prover is run on $\gamma$-noisy devices for $\gamma\leq \gamma_0$, the verifier still accepts with probability $1-\negl(n)$.
\item Soundness: Any classical prover (which runs in polynomial time or not), which transmits a classical memory register $M$ satisfying $\sfI^{\negl(n)}_{\max}(X:M) \leq s(n)$ from $t_0$ to $t_1$, will make the verifier accept with $\negl(n)$ probability.
\end{itemize}
\end{definition}

In the above definition, we intentionally do not specify the noise model, since it is a practical consideration, and different protocols that are implemented on different hardware will have different types of noise. Nevertheless, since any implementation of a quantum protocol will have noise, noise robustness is a desirable property; our intent is for the relevant type of noise model to be used for each type of protocol in this definition.

Our result with respect to our definition of quantum information advantage is stated formally in Theorem \ref{thm:main}.
\begin{restatable}{theorem}{main}
\label{thm:main}
There exists a $(0, \Omega(n), 0.01)$ noise-robust proof of quantum information advantage. Moreover, this protocol satisfies $\sfI(X:M) = \sfI_{\max}^0(X:M)=0$. 

Let $\delta:=2^{-n/((\ln2)\times 10^4)}$. Any classical prover whose classical memory satisfies $$\Imax^\delta(X:M) \leq \frac{4n}{
(\ln 2) \times 10^4} - \log \log \frac{1}{\delta} - 19,$$ will make the verifier accept with probability at most $72\delta^2$.
 \end{restatable}
\paragraph{Protocol description.} Our protocol for proof of quantum information advantage is based on a Bell inequality violation, specifically of the CHSH game. Bell inequality violations can be turned into separations in communication for relational problems by requiring that Alice and Bob produce outputs such that the fraction of correct outputs according to the nonlocal game winning conditions is close to the quantum winning probability of the game. Such a non-local game is called a threshold parallel-repeated game. We consider the $t$-out-of-$n$ threshold parallel repeated game $\chsh^{t/n}$, whose input and output sets are all $n$-bit strings. The winning condition of the game is that the strings $x_1\ldots x_n, y_1\ldots y_n, a_1\ldots a_n, b_1\ldots b_n$ satisfy
\[ \left|\{i: a_i\oplus b_i=x_i\cdot y_i\}\right| \geq t,\]
with $t=0.83n$. Due to being based on non-local games, compared to \cite{KGD+25},  our proposal has the advantage of requiring states and measurements that are far less complicated. The protocol is pictorially depicted in Figure~\ref{fig:circuit}.
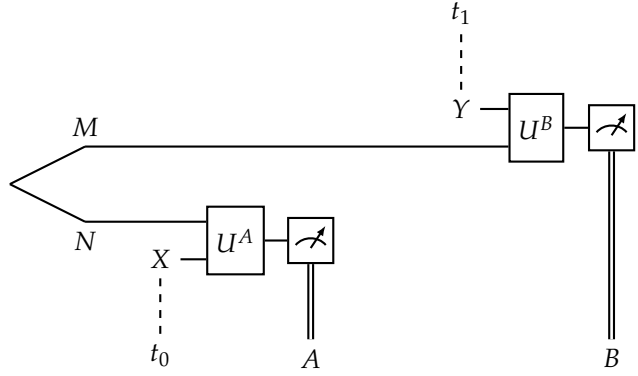
\begin{figure}[!ht]
\centering
\begin{tikzpicture}[thick]

\coordinate (fork) at (0,0);
\coordinate (M) at (1,0.5);
\coordinate (Nstart) at (1,-0.5);

\draw (fork) -- (M) node[right, above] {$M$};
\draw (fork) -- (Nstart) node[right, below] {$N$};

\node[gate] (UA) at (3,-0.75) {$U^A$};
\node (X) at (2,-1) {$X$};

\draw (X) -- (UA.west |- X);
\draw (Nstart) -- (UA.west |- Nstart);

\node[meter,draw,minimum size=6mm] (mA) at (4,-0.75) {};
\draw (UA.east |- mA) -- (mA);

\draw[cwire] (mA.south) -- ++(0,-1) node[below] {$A$};

\draw[dashed] (X) -- ++(0,-1) node[below] {$t_0$};

\node[gate] (UB) at (7,0.75) {$U^B$};

\node (Y) at (6,1) {$Y$};
\draw (Y) -- (UB.west |- Y);
\draw (M) -- (UB.west |- M);

\node[meter,draw,minimum size=6mm] (mB) at (8,0.75) {};
\draw (UB.east |- mB) -- (mB);

\draw[cwire] (mB.south) -- ++(0,-2.5) node[below] {$B$};

\draw[dashed] (Y) -- ++(0,1) node[above] {$t_1$};

\end{tikzpicture}

\caption{Circuit implemented by the quantum prover in our proof of quantum information advantage protocol. The unitaries $U^A$ and $U^B$ are Alice and Bob's unitaries in the optimal quantum strategy for the $\chsh^n$ game. The register $M$ is the quantum memory.}
\label{fig:circuit}
\end{figure}

\paragraph{Noise model.} We now clarify what is meant by noise in our setting. If the quantum value of a non-local game is $\omega^*$, we refer to quantum devices that are capable of achieving value $\omega^*-\gamma$, $\gamma$-noisy devices. We do not seek to model the inner workings of the device that leads to this value --- is consistent with how noise is modeled in many Bell scenarios, see for example, \cite{Nad22, Zha22}. There is some constant level of noise in each single-qubit gate which leads to this value of $\gamma$.
\paragraph{Concrete parameters.} 
\begin{enumerate}
    \item Let $n=1.1697\times 10^4$ and the smoothing parameter $\delta=\frac{1}{12}$ (here $n$ and $\delta$ do not have the same relation as in Theorem~\ref{thm:main}, see Theorem~\ref{thm:cl-inf-ub}).
    \item The success probability of the quantum prover (with $\sfI(X:M) = \sfI_{\max}^0(X:M)=0$) in our protocol is at least $1-e^{-2n/10^4}=0.903$ (see the proof of Theorem~\ref{thm:main}). 
    \item The success probability of a classical prover who has classical memory $\Imax^{1/12}(X:M) \leq 100$, is at most $72\delta^2=\frac{1}{2}$. 
\end{enumerate}
Thus in order to achieve a $0$ vs $100$ quantum-classical separation, we need to be able to do the CHSH winning strategy in parallel on $1.1697\times 10^4$ many copies, with $\gamma=0.01$ noise.

\subsection{Comparison with previous work}
We now give a more detailed comparison between our proposal and the work of \cite{KGD+25}. The task used to prove quantum information advantage in \cite{KGD+25} is a version of distributed linear cross-entropy benchmarking. In this task, at time $t_0$, the prover gets the classical description of a Haar-random $m$-qubit quantum state $\ket{\psi}$ and at time $t_1$ they get the description of a Haar-random unitary $U$, and the requirement is to output an $m$-bit string $z_{\psi,U}$ such that
\[ \bbE_{\psi, U}\left[\left|\bra{z_{\psi,U}}U\ket{\psi}\right|^2\right] \geq \frac{1+\eps}{2^m}.\]
Note that an $m$-qubit Haar-random quantum state takes $\widetilde{O}(2^m)$ many bits to describe, so the input parameter here is $n=2^m$. \cite{KGD+25} use $\eps$ to characterize the correctness of their protocol: it is trivial to achieve $\eps=0$. They show that a quantum prover using $m$ qubits of quantum memory can achieve $\eps=1$, whereas a classical prover using $s$ bits of classical memory satisfies
\[ \eps = O\left(\max\left\{\sqrt{\frac{s}{2^m}}, \frac{sm}{2^m}\right\}\right).\]
The definition of $\eps$ for both the classical and quantum prover here is inherently average-case. We also note that the unitary at time $t_1$ can come from a number of smaller sets of unitaries which are efficiently implementable as well, such as $2$-designs or products of single-qubit Clifford gates; they show similar upper bounds on $\eps$ for $U$ coming from these sets.

We note the following points of comparison between our proposal and the proposal in~\cite{KGD+25}.
\begin{itemize}
\item \textbf{Relation:} The relation we use has two outputs, $a$ and $b$. The sampling problem considered in \cite{KGD+25} has two inputs $x$ and $y$ like ours, but only one output that has to be produced after $t_1$ when $y$ is given. For us, it is important that the output $a$ is produced when $x$ is given at $t_0$ and only output $b$ is produced later. However, we do not consider this feature of our proposal much of a drawback, since having two outputs is not particularly unnatural compared to having one output.
\item \textbf{Memory measure:} The information measure used in \cite{KGD+25} is the total number of qubits or bits that need to be passed from $t_0$ to $t_1$.

For us it is instead an information measure about the input $x$ in these bits or qubits. Notably, the number of bits or qubits respectively that need to be passed from $t_0$ to $t_1$ is $\Omega(n)$ for us in the classical and quantum cases both, it is only that in the quantum case the qubits carry no information about $x$.

\item \textbf{Completeness and soundness errors:} The notion of completeness and soundness errors we use is much more standard. 

As far as we understand, the protocol in \cite{KGD+25} is not explicitly presented as an interactive proof with standard notions of completeness and soundness errors. This is because theirs is a sampling problem, and the verifier does not accept or reject based on some condition, and instead computes the value of $\eps$. 
\item \textbf{Verifier's efficiency:} Our verifier can straightforwardly check in $O(n)$ time whether the CHSH winning condition is satisfied in $0.83$ fraction of $i$-s for each input and output. 

In \cite{KGD+25}, a classical verifier can calculate $\left|\bra{z}U\ket{\psi}\right|^2$ for a given $U$ and $\ket{\psi}$ in $\poly(2^m)$ time, which is efficient in $n=2^m$. But as far as we understand, the verifier has to run the linear cross-entropy protocol several times and find an estimate of $\eps$ from this. This would require the verifier to send several $\ket{\psi}$-s at $t_0$ and several $U$-s at $t_1$. The number of $\ket{\psi}$-s and $U$-s (that would depend on $\eps$) to be sent does not appear to be explicitly stated in~\cite{KGD+25}. Also it is not clear if the lower bound for classical memory would hold for this parallel repeated protocol. 
\item \textbf{Quantum prover's efficiency:} The state preparation and measurement procedures for us are much simpler than those in \cite{KGD+25}. We only require the preparation of EPR pairs, and $4$ types of single-qubit gates (corresponding to the measurements in the CHSH non-local game), which should be possible to implement in existing quantum devices. 

\cite{KGD+25} can use $2$-designs or random Cliffords for $U$, but they do actually need to prepare an $m$-qubit Haar-random $\ket{\psi}$. Exact Haar-random need $\Omega(2^m)$ gates and $\Omega(2^m/m)$ depth to prepare, which may have been infeasible in their actual implementation. They use a brickwork ansatz with much smaller circuit size and depth in its place, in their actual implementation. However, as far as we understand, their classical information lower bound does not actually hold against this ansatz, since this ansatz has a $\poly(m)$-size description.
\item \textbf{Noise-robustness:} \cite{KGD+25} do not do an asymptotic noise robustness analysis for their quantum protocol. While it may be possible to get a nontrivial output for their brickwork ansatz (which does not have a classical lower bound) if each gate has constant error, their actual lower bound requires Haar-random states. If each of the $\Omega(2^n)$ gates required to produce the Haar-random state is implemented with constant error, we suspect their quantum protocol would not achieve a nontrivial value of $\eps$. 

Our protocol in contrast is robust to each gate having constant noise in such a way that the single-copy CHSH winning probability is still $0.84$.

\item \textbf{Concrete numbers for quantum memory:} The actual number of qubits required to achieve a $0$ vs $100$ quantum information advantage in our proposal is of the order of $10^4$. However, doing only CHSH measurements on $10^4$ EPR pairs is possible with current technology, and has already been done with the kind of error rate we require in the context of demonstrating Bell inequality violations and device-independent cryptography \cite{Nad22, Zha22}. We note however that the measurements in these device-independent protocols were done sequentially, and some spatial separation needed to be maintained between two sets of devices. Our measurements need to be done in parallel, but we don't need this spatial separation.

As we have stated before, the exact memory needed by quantum prover in the protocol by~\cite{KGD+25} is not clear since the number of times $(\ket{\psi}, U)$-s that have to be sent is not stated. They have shown that a single repetition requires only $12$ qubits.
\end{itemize}

\section{Preliminaries}\label{sec:prelims}
\subsection{Tail bounds}
\begin{fact}[Chernoff bound] \label{fact:chernoff}
Let $X_1, \ldots, X_n$ be i.i.d. Bernoulli variables with $\bbE[X_i] = \delta$. Then,
\begin{align*}
\Pr\left[\sum_{i=1}^n X_i \geq \gamma n \right] & \leq e^{-2(\gamma-\delta)^2n}; \\
\Pr\left[\sum_{i=1}^n X_i \leq \gamma n \right] & \leq e^{-2(\gamma-\delta)^2n}.
\end{align*}
\end{fact}

\begin{fact}[Generalized Chernoff bound, \cite{PS97}]\label{fc:gen-chernoff}
Let $X_1, \ldots, X_n$ be boolean random variables such that for some $\delta \in [0,1]$, we have for every subset $S \subseteq [n]$, $\Pr\left[\land_{i\in S}X_i\right] \leq \delta^{|S|}$. Then, for any $\gamma \in [\delta, 1]$, we have,
\[\Pr\left[\sum_{i=1}^nX_i \geq \gamma n\right] \leq e^{-2(\gamma-\delta)^2n}.\]
\end{fact}

\subsection{Information measures}
Let $XY$ be two random variables with the joint distribution $p_{XY}$. We use $p_{Y|x}$ to denote the distribution of $Y$ conditioned on $X=x$.
\begin{definition}[Smoothed max-information]\label{def:Imax}
For $\gamma \in [0,1]$, the $\gamma$-smoothed maximum mutual information $\Imax$ between registers $X$ and $Y$ in the state $\rho_{XY}$ is given by
\[ \Imax^\gamma(X:Y)_\rho = \inf\left\{c : \, \rho'_{XY} \preceq 2^c \rho_X \otimes \rho_Y, \, \frac{1}{2}\Vert \rho'_{XY} - \rho_{XY}\Vert_1 \leq \gamma \right\}. \]
When $X$ and $Y$ are both classical, and they have the joint classical distribution $p_{XY}$, we omit the subscript in the notation, and the quantity reduces to
\[ \Imax^\gamma(X:Y) = \inf\left\{c: p'_{XY}(xy) \leq 2^c p(x)p(y) \, \forall x, y, \, \frac{1}{2}\Vert p'_{XY} - p_{XY}\Vert_1 \leq \gamma \right\}.\]
\end{definition}
\begin{definition}[Classical smoothed spectrum max-information, \cite{ABJT18}]
For classical random variables $XY$, the $\gamma$-smoothed spectrum maximum mutual information between them is defined as
\[ \Is^\gamma(X:Y) = \inf\left\{c: \Pr_{(x,y) \sim p_{XY}}\left[\frac{p_{Y|x}(y)}{p_Y(y)}>2^c\right]\leq \gamma\right\}.\]
\end{definition}
For both $\Imax^\gamma$ and $\Is^\gamma $, it is not difficult to see that they are increasing functions of $\gamma$. We use simply $\Imax(X:Y)_\rho$ and $\Is(X:Y)$ for $\Imax^0(X:Y)_\rho$ and $\Is^0(X:Y)$.

Moreover, if two variables $X$ and $Y$ are uncorrelated, then $\Imax(X:Y)=\Is(Y:X)=0$.

The two maximum mutual information quantities for classical variables can be shown to satisfy the following relation.
\begin{fact}[\cite{ABJT18}, Theorem 1 and Theorem 2]\label{fc:Is-Imax}
For $\gamma \leq 1/\sqrt{12}$,
\[\Is^{18\gamma^2}(X:Y) \leq \Imax^\gamma(X:Y) + \log\left(\frac{8+\gamma^2}{\gamma^2}\right) + 2.\]
\end{fact}

A more standard measure of information between two variables or quantum registers is the mutual information, defined below.
\begin{definition}\label{def:I}
The mutual information between registers $X$ and $Y$ in a quantum state $\rho_{XY}$ is defined as
\[ \sfI(X:Y)_\rho = \sfH(X)_\rho + \sfH(Y)_\rho - \sfH(XY)_\rho,\]
where $\sfH(A)_\rho = - \Tr(\rho_A\log\rho_A)$. For a classical random variable $X$, $\sfH(X)$ reduces to $\sum_xp_X(x)\log(1/p_X(x))$.
\end{definition}
We always have $\sfI(X:Y)_\rho \leq \Imax(X:Y)_\rho$, and in particular if $\Imax(X:Y)=0$, then $\sfI(X:Y)=0$.

\subsection{Nonlocal games}
A two-player non-local game $G$ is described as $(p, \clX\times\clY,\clA\times\clB, \sfV)$ where $p$ is a probability distribution over the input set $\clX\times\clY$, $\clA\times\clB$ is the output set, and $\sfV:\clX\times\clY\times\clA\times\clB \to \{0,1\}$ is a predicate. It is played as follows: a referee selects inputs $(x,y)$ according to $p$, sends $x$ to Alice and $y$ to Bob. In a classical strategy for $G$, Alice and Bob send outputs $a$ and $b$ back to the referee, which are just functions of $x$ and $y$~\footnote{In principle they can also share randomness, and produce outputs that are functions of their inputs as well as the shared randomness, but it can be shown that this does not help.} In a quantum strategy, Alice and Bob are allowed to share entanglement; they perform measurements on their respective halves of the entangled state depending on their inputs, and send their outputs $(a,b)$ back to the referee. The referee accepts and Alice and Bob win the game iff $\sfV(x,y,a,b) = 1$.
\begin{definition}
The classical value of a game $G=(p,\clX\times\clY,\clA\times\clB,\sfV)$, denoted by $\omega(G)$, is the maximum winning probability of Alice and Bob, averaged over the distribution $p$, over all classical strategies for $G$.

The quantum value of a game $G$, denoted by $\omega^*(G)$, is the maximum winning probability of Alice and Bob, averaged over the distribution $q$ as well as inherent randomness in the strategy, over all quantum strategies for $G$.
\end{definition}

For any game $G$, we can define its parallel-repeated version, and its threshold parallel-repeated version, as follows.
\begin{itemize}
\item In the standard parallel-repeated version of $G$, denoted by $G^n$, Alice and Bob get inputs $x_1, \ldots, x_n$ and $y_1\ldots y_n$ from $\clX^n\times\clY^n$ according to the product input distribution $p^{\otimes n}$, and must produce outputs $a_1\ldots a_n, b_1\ldots b_n \in \clA^n\times\clB^n$ such that for all $i\in[n]$, $\sfV(x_i, y_i, a_i, b_i)$ is satisfied.
\item In the $t$-out-of-$n$ threshold version of $G$, denoted by $G^{t/n}$, Alice and Bob get the same inputs as in $G^n$, and must produce outputs from $\clA^n\times\clB^n$, such that their input and outputs satisfy
\[ \left|\left\{i\in [n]: \sfV(x_i, y_i, a_i, b_i)=1\right\}\right| \geq t.\]
\end{itemize}

Of particular interest to us will be the CHSH game. In a single copy of the CHSH game, the inputs and outputs are all bits, and they must satisfy $a\oplus b = x\cdot y$. The classical value of a single copy of CHSH is $\frac{3}{4}$, while the quantum value is $\cos^2(\pi/8) \approx 0.85$.

The optimal strategy for a parallel-repeated game is not always to do the optimal strategy for the single-copy game $n$ times. The exact values of parallel-repeated games are therefore often not known. However, for the CHSH game, the parallel-repeated classical and quantum values are in fact well-characterized.
\begin{fact}[\cite{CSUU08}]
The quantum value of $n$ copies of the CHSH game is
\[ \omega^*(\chsh^n) = \cos^{2n}\left(\frac{\pi}{8}\right).\]
\end{fact}

The following upper bound on the classical value of $\chsh^n$ is an unpublished result due to Ambainis, following \cite{DS14}.
\begin{fact}[\cite{DS14}, see also \cite{SO14}]\label{thm:chsh-n}
The classical value of $n$ copies of the CHSH game satisfies
\[ \omega(\chsh^n) \leq \left(\frac{1+\sqrt{5}}{4}\right)^n \approx (0.81)^n.\]
Moreover, for large enough $n$, the inequality is an equality.
\end{fact}

\section{Task achieving quantum information advantage}\label{sec:main}
In this section, we prove Theorem \ref{thm:main}, which is restated below.
\main*

\begin{proof}[Proof of Theorem \ref{thm:main}]
As already stated, we use $\chsh^{t/n}$ as the relation for the proof of quantum information advantage protocol. The protocol using this relation can be described as follows:
\begin{enumerate}
\item The verifier and the prover get $1^n$.
\item At time $t_0$, the verifier samples uniform $x\in\{0,1\}^n$ and sends it to the prover.
\item The prover replies with $a\in\{0,1\}^n$, which the verifier stores.
\item At time $t_1$, the verifier samples uniform $y\in\{0,1\}^n$ and sends it to the prover.
\item The prover replies with $b\in\{0,1\}^n$.
\item The verifier accepts if $x,y,a,b$ satisfy
\[ \left|\{i: a_i\oplus b_i=x_i\cdot y_i\}\right| \geq  0.83n.\]
\end{enumerate}
For soundness of the protocol, we invoke Theorem \ref{thm:cl-inf-ub} with $\eps = \frac{2}{100}$ and $\delta=\frac{n}{(\ln 2)\times 10^4}$. This implies an upper bound $72\delta^2$ on the winning probability of a classical prover satisfying 
$$\Imax^\delta(X:M) \leq \frac{4n}{(\ln 2)\times 10^4} - \log\log \frac{1}{\delta}- 19.$$ 

We will now prove the completeness and noise-robustness part of the theorem. The quantum strategy for $\chsh^{t/n}$ will just be running the optimal quantum strategy for CHSH $n$ times, with the caveat that Alice's measurement is done at time $t_0$ and Bob's measurement at time $t_1$. In the optimal quantum strategy for CHSH, Alice and Bob share an EPR pair, and the optimal quantum winning probability is $\cos^2\left(\frac{\pi}{8}\right) \approx 0.85$. If instead of $n$ exact EPR pairs, in our model Alice and Bob have access to $n$ noisy EPR pairs, and on each copy they can perform noisy versions of the measurements corresponding to the CHSH strategy such that the winning probability of a single copy of CHSH is at least $0.85 - \gamma$. Since winning on each copy is independent, by the Chernoff bound (Fact~\ref{fact:chernoff}), the probability that the winning condition is satisfied on at least $(0.85-2\gamma)n = 0.83n$ copies is at least $1 - e^{-2\gamma^2n}$.

The quantum memory $M$ that needs to be preserved between times $t_0$ and $t_1$ is Bob's part of the noisy EPR pairs. By no-signaling, Bob's system is not correlated with Alice's input $X$. Therefore, we have $\Imax(X:M) = 0$.
\end{proof}

\subsection{Classical lower bound}
We first prove an upper bound on the classical value of the threshold parallel-repeated CHSH game in the standard non-local game setting without communication, and then extend it to the communication setting.
\begin{lemma}\label{lem:CHSH-thres}
The classical value of winning $t$ out of $n$ CHSH games, for $t=((1+\sqrt{5})/4+\eps)n$ is upper bounded by
\[ \omega(\chsh^{t/n}) \leq e^{-2\eps^2n}. \]
\end{lemma}


\begin{proof}
First we argue that for $n$ CHSH games, the probability of winning in any subset $S$ is at most $\left(\frac{1+\sqrt{5}}{4}\right)^{|S|}$. Suppose this is not true for some subset $S$. Then we can give a strategy for $\chsh^{|S|}$ as follows: use the strategy for $\chsh^n$, embed the actual $|S|$ inputs for $\chsh^{|S|}$ into the subset $S$ of $n$, and sample the $n-|S|$ other inputs according to the CHSH distribution. The probability of winning $\chsh^{|S|}$ using this strategy is exactly the probability of winning in the subset $S$ of the original strategy for $\chsh^n$. The probability of winning in $S$ thus cannot be higher than $\left(\frac{1+\sqrt{5}}{4}\right)^{|S|}$ by Fact \ref{thm:chsh-n}.

From the above argument, and using the generalized Chernoff bound (Fact \ref{fc:gen-chernoff}), we have the desired result.
\end{proof}

\begin{cor}\label{cor:cl-comm-ub}
Suppose Alice and Bob, who share randomness but no entanglement, get inputs for $n$ copies of CHSH. Alice produces her outputs, sends a classical message to Bob consisting of $c$ bits, and Bob produces his outputs using this classical message from Alice. Then their probability of winning $t$ out of $n$ copies of CHSH for $t=((1+\sqrt{5})/4+\eps)n$ satisfies
\[ \Pr\left[\text{Win } \chsh^{t/n}\right] \leq 2^c\cdot e^{-2\eps^2n}.\]
\end{cor}

\begin{proof}
Suppose the upper bound does not hold. Then we can give a strategy for $\chsh^{t/n}$ without communication with winning probability more than $e^{-2\eps^2n}$, contradicting Lemma \ref{lem:CHSH-thres}. The way this strategy works is that Bob will pretend his private randomness is the message from Alice. Since Alice's message is $c$ bits, Bob's private randomness will be equal to the message Alice would have sent with probability $2^{-c}$. Therefore, if the probability of winning with the communication is greater than $2^c\cdot e^{-2\eps^2n}$, then the probability of winning without communication is greater than $2^{-c}\cdot 2^c\cdot e^{-2\eps^2n} = e^{-2\eps^2n}$.
\end{proof}

We are now ready to prove the main result of this section, where we characterize the communication by the information it conveys about $X$, and relate this to the winning probability of $\chsh^{t/n}$.
\begin{theorem}\label{thm:cl-inf-ub}
Suppose Alice and Bob, who share randomness but no entanglement, get inputs $X$ and $Y$ for $n$ copies of CHSH. Alice produces her outputs, sends a classical message $M$ satisfying 
$$\Imax^{\delta}(X:M) < \frac{2\eps^2n}{\ln2} - 4\log(1/\delta) - \log\log(1/\delta) - 19$$ for some $\delta, \eps >0$, and Bob produces his outputs using this classical message from Alice. Then their probability of winning $t$ out of $n$ copies of CHSH for $t=((1+\sqrt{5})/4+\eps)n$ satisfies
\[ \Pr\left[\text{Win } \chsh^{t/n}\right] < 72\delta^2.\]
\end{theorem}

In order to prove this theorem, we use the following result on compressing classical messages.
\begin{fact}[\cite{JRS03,  Jai11}]\label{thm:cl-comp}
Suppose $M$ is the transcript of a one-way private coin classical communication protocol $\clP$ between Alice and Bob, who have inputs $X$ and $Y$ respectively from a product distribution, which has success probability $\alpha$ for some relation (under the product distribution on $X, Y$). If $\Is^\beta(X:M) \leq c$, for $\beta \in (0,\alpha/2]$, then there exists another deterministic one-way communication protocol $\clP'$ between Alice and Bob such that:
\begin{enumerate}[(i)]
\item The communication in $\clP'$ is at most $c+\log(2/\beta)+\log\log(2/\beta)$ bits;
\item The success probability of $\clP'$ for the relation is at least $\alpha - 2\beta$.
\end{enumerate}
\end{fact}

\begin{proof}[Proof of Theorem \ref{thm:cl-inf-ub}]
Let $\delta'=72\delta^2$. Suppose the winning probability of the original protocol for $\chsh^{t/n}$ (which we call $\clP$) is at least $\delta'$. We'll prove a lower bound on $ c = \Is^{\delta'/4}(X:M)$, which gives us the result via Fact \ref{fc:Is-Imax}. If 
$$c < \frac{2\eps^2n}{\ln2} - \log(16/\delta') -\log\log(8/\delta')$$ 
in $\clP$, then by Fact~\ref{thm:cl-comp}, there exists another protocol $\clP'$ whose success probability is at least $\delta' - 2\cdot\frac{\delta'}{4}$, and where Alice communicates 
$$c' = c + \log(8/\delta') + \log\log(8/\delta')$$ 
bits to Bob. By Corollary \ref{cor:cl-comm-ub}, we must have,
\[ \frac{\delta'}{2} \leq 2^{c'}\cdot e^{-2\eps^2n} \quad \quad \Rightarrow \quad \quad c' \geq \frac{2\eps^2n}{\ln 2} - \log(2/\delta').\]
Putting in the value of $c'$, we get,
\begin{align*}
c & \geq \frac{2\eps^2n}{\ln 2} - \log(2/\delta') - \log(8/\delta') - \log\log(8/\delta'), \\
\end{align*}
which contradicts the upper bound on $c$.

Putting in the value of $\delta'$, we get from Fact \ref{fc:Is-Imax},
\begin{align*}
\Imax^\delta(X:M) & \geq \frac{2\eps^2n}{\ln 2} - \log(1152/\delta^2) - \log\log(576/\delta^2) - \log((8+\delta^2)/\delta^2) - 2 \\
 & \geq \frac{2\eps^2n}{\ln 2} - 4\log(1/\delta) - \log\log(1/\delta) - 19
\end{align*}
for $\delta \leq \frac{1}{2}$.
\end{proof}

\bibliographystyle{alpha}
\bibliography{ref}

@InProceedings{JRS03,
author="Jain, Rahul
and Radhakrishnan, Jaikumar
and Sen, Pranab",
title="A Direct Sum Theorem in Communication Complexity via Message Compression",
booktitle="International Colloquium on Automata, Languages and Programming (ICALP 2003)",
year="2003",
pages="300--315",
isbn="978-3-540-45061-0"
}

@article{Jai11,
author = {Jain, Rahul},
title = {New Strong Direct Product Results in Communication Complexity},
year = {2015},
volume = {62},
number = {3},
issn = {0004-5411},
url = {https://doi.org/10.1145/2699432},
doi = {10.1145/2699432},
journal = {Journal of the ACM},
articleno = {20},
numpages = {27}
}

@article{PS97,
author = {Panconesi, Alessandro and Srinivasan, Aravind},
title = {Randomized Distributed Edge Coloring via an Extension of the Chernoff--Hoeffding Bounds},
journal = {SIAM Journal on Computing},
volume = {26},
number = {2},
pages = {350-368},
year = {1997},
doi = {10.1137/S0097539793250767},
URL = {https://doi.org/10.1137/S0097539793250767},
eprint = {https://doi.org/10.1137/S0097539793250767},
}

@inproceedings{DS14,
author = {Dinur, Irit and Steurer, David},
title = {Analytical approach to parallel repetition},
year = {2014},
isbn = {9781450327107},
url = {https://doi.org/10.1145/2591796.2591884},
doi = {10.1145/2591796.2591884},
booktitle = {Proceedings of the Forty-Sixth Annual ACM Symposium on Theory of Computing (STOC 2014)},
pages = {624–633},
numpages = {10},
}

@misc{KGD+25,
title={Demonstrating an unconditional separation between quantum and classical information resources}, 
      author={William Kretschmer and Sabee Grewal and Matthew DeCross and Justin A. Gerber and Kevin Gilmore and Dan Gresh and Nicholas Hunter-Jones and Karl Mayer and Brian Neyenhuis and David Hayes and Scott Aaronson},
      year={2025},
      eprint={2509.07255},
      archivePrefix={arXiv},
      primaryClass={quant-ph},
      howpublished={\url{https://arxiv.org/abs/2509.07255}}
}

@misc{SO14,
author = {Scott Aaronson},
title = {{The NEW Ten Most Annoying Questions in Quantum Computing}},
year = {2014},
howpublished = {\url{https://scottaaronson.blog/?p=1792}},
}

@ARTICLE{ABJT18,
  author={Anshu, Anurag and Berta, Mario and Jain, Rahul and Tomamichel, Marco},
  journal={IEEE Transactions on Information Theory}, 
  title={Partially Smoothed Information Measures}, 
  year={2020},
  volume={66},
  number={8},
  pages={5022-5036},
  doi={10.1109/TIT.2020.2981573}
}

@article{Aru19,
  author  = {Frank Arute and Kunal Arya and Ryan Babbush and Dave Bacon and Joseph C. Bardin and Rami Barends and Rupak Biswas and Sergio Boixo and Fernando G. S. L. Brandao and David A. Buell and Brian Burkett and Yu Chen and Zijun Chen and Ben Chiaro and Roberto Collins and William Courtney and Andrew Dunsworth and Edward Farhi and Brooks Foxen and Austin Fowler and Craig Gidney and Marissa Giustina and Rob Graff and Keith Guerin and Steve Habegger and Matthew P. Harrigan and Michael J. Hartmann and Alan Ho and Markus Hoffmann and Trent Huang and Travis S. Humble and Sergei V. Isakov and Evan Jeffrey and Zhang Jiang and Dvir Kafri and Kostyantyn Kechedzhi and Julian Kelly and Paul V. Klimov and Sergey Knysh and Alexander Korotkov and Fedor Kostritsa and David Landhuis and Mike Lindmark and Erik Lucero and Dmitry Lyakh and Salvatore Mandrà and Jarrod R. McClean and Matthew McEwen and Anthony Megrant and Xiao Mi and Kristel Michielsen and Masoud Mohseni and Josh Mutus and Ofer Naaman and Matthew Neeley and Charles Neill and Murphy Yuezhen Niu and Eric Ostby and Andre Petukhov and John C. Platt and Chris Quintana and Eleanor G. Rieffel and Pedram Roushan and Nicholas C. Rubin and Daniel Sank and Kevin J. Satzinger and Vadim Smelyanskiy and Kevin J. Sung and Matthew D. Trevithick and Amit Vainsencher and Benjamin Villalonga and Theodore White and Z. Jamie Yao and Ping Yeh and Adam Zalcman and Hartmut Neven and John Matthew Martinis},
  title   = {{Quantum supremacy using a programmable superconducting processor}},
  journal = {Nature},
  volume  = {574},
  pages   = {505--510},
  year    = {2019},
  doi     = {10.1038/s41586-019-1666-5}
}

@article{Wu21,
  author  = {Yulin Wu and Wan-Su Bao and Sirui Cao and Fusheng Chen and Ming-Cheng Chen and Xiawei Chen and Tung-Hsun Chung and Hui Deng and Yajie Du and Daojin Fan and Ming Gong and Cheng Guo and Chu Guo and Shaojun Guo and Lianchen Han and Linyin Hong and He-Liang Huang and Yong-Heng Huo and Liping Li and Na Li and Shaowei Li and Yuan Li and Futian Liang and Chun Lin and Jin Lin and Haoran Qian and Dan Qiao and Hao Rong and Hong Su and Lihua Sun and Liangyuan Wang and Shiyu Wang and Dachao Wu and Yu Xu and Kai Yan and Weifeng Yang and Yang Yang and Yangsen Ye and Jianghan Yin and Chong Ying and Jiale Yu and Chen Zha and Cha Zhang and Haibin Zhang and Kaili Zhang and Yiming Zhang and Han Zhao and Youwei Zhao and Liang Zhou and Qingling Zhu and Chao-Yang Lu and Cheng-Zhi Peng and Xiaobo Zhu and Jian-Wei Pan},
  title   = {{Strong quantum computational advantage using a superconducting quantum processor}},
  journal = {Physical Review Letters},
  volume  = {127},
  pages   = {180501},
  year    = {2021},
  doi     = {10.1103/PhysRevLett.127.180501}
}

@article{Zhu22,
  author  = {Qingling Zhu and Sirui Cao and Fusheng Chen and Ming-Cheng Chen and Xiawei Chen and Tung-Hsun Chung and Hui Deng and Yajie Du and Daojin Fan and Ming Gong and Cheng Guo and Chu Guo and Shaojun Guo and Lianchen Han and Linyin Hong and He-Liang Huang and Yong-Heng Huo and Liping Li and Na Li and Shaowei Li and Yuan Li and Futian Liang and Chun Lin and Jin Lin and Haoran Qian and Dan Qiao and Hao Rong and Hong Su and Lihua Sun and Liangyuan Wang and Shiyu Wang and Dachao Wu and Yulin Wu and Yu Xu and Kai Yan and Weifeng Yang and Yang Yang and Yangsen Ye and Jianghan Yin and Chong Ying and Jiale Yu and Chen Zha and Cha Zhang and Haibin Zhang and Kaili Zhang and Yiming Zhang and Han Zhao and Youwei Zhao and Liang Zhou and Chao-Yang Lu and Cheng-Zhi Peng and Xiaobo Zhu and Jian-Wei Pan},
  title   = {{Quantum computational advantage via 60-qubit 24-cycle random circuit sampling}},
  journal = {Science Bulletin},
  volume  = {67},
  pages   = {240--245},
  year    = {2022},
  doi     = {10.1016/j.scib.2021.10.017}
}

@article{Mor24,
  author  = {A. Morvan and B. Villalonga and X. Mi and S. Mandrà and A. Bengtsson and P. V. Klimov and Z. Chen and S. Hong and C. Erickson and I. K. Drozdov and J. Chau and G. Laun and R. Movassagh and A. Asfaw and L. T.A.N. Brandão and R. Peralta and D. Abanin and R. Acharya and R. Allen and T. I. Andersen and K. Anderson and M. Ansmann and F. Arute and K. Arya and J. Atalaya and J. C. Bardin and A. Bilmes and G. Bortoli and A. Bourassa and J. Bovaird and L. Brill and M. Broughton and B. B. Buckley and D. A. Buell and T. Burger and B. Burkett and N. Bushnell and J. Campero and H. S. Chang and B. Chiaro and D. Chik and C. Chou and J. Cogan and R. Collins and P. Conner and W. Courtney and A. L. Crook and B. Curtin and D. M. Debroy and A. {Del Toro Barba} and S. Demura and A. {Di Paolo} and A. Dunsworth and L. Faoro and E. Farhi and R. Fatemi and V. S. Ferreira and L. {Flores Burgos} and E. Forati and A. G. Fowler and B. Foxen and G. Garcia and E. Genois and W. Giang and C. Gidney and D. Gilboa and M. Giustina and R. Gosula and A. {Grajales Dau} and J. A. Gross and S. Habegger and M. C. Hamilton and M. Hansen and M. P. Harrigan and S. D. Harrington and P. Heu and M. R. Hoffmann and T. Huang and A. Huff and W. J. Huggins and L. B. Ioffe and S. V. Isakov and J. Iveland and E. Jeffrey and Z. Jiang and C. Jones and P. Juhas and D. Kafri and T. Khattar and M. Khezri and M. Kieferová and S. Kim and A. Kitaev and A. R. Klots and A. N. Korotkov and F. Kostritsa and J. M. Kreikebaum and D. Landhuis and P. Laptev and K. M. Lau and L. Laws and J. Lee and K. W. Lee and Y. D. Lensky and B. J. Lester and A. T. Lill and W. Liu and W. P. Livingston and A. Locharla and F. D. Malone and O. Martin and S. Martin and J. R. McClean and M. McEwen and K. C. Miao and A. Mieszala and S. Montazeri and W. Mruczkiewicz and O. Naaman and M. Neeley and C. Neill and A. Nersisyan and M. Newman and J. H. Ng and A. Nguyen and M. Nguyen and M. {Yuezhen Niu} and T. E. O'Brien and S. Omonije and A. Opremcak and A. Petukhov and R. Potter and L. P. Pryadko and C. Quintana and D. M. Rhodes and E. Rosenberg and C. Rocque and P. Roushan and N. C. Rubin and N. Saei and D. Sank and K. Sankaragomathi and K. J. Satzinger and H. F. Schurkus and C. Schuster and M. J. Shearn and A. Shorter and N. Shutty and V. Shvarts and V. Sivak and J. Skruzny and W. C. Smith and R. D. Somma and G. Sterling and D. Strain and M. Szalay and D. Thor and A. Torres and G. Vidal and C. {Vollgraff Heidweiller} and T. White and B. W. K. Woo and C. Xing and Z. J. Yao and P. Yeh and J. Yoo and G. Young and A. Zalcman and Y. Zhang and N. Zhu and N. Zobrist and E. G. Rieffel and R. Biswas and R. Babbush and D. Bacon and J. Hilton and E. Lucero and H. Neven and A. Megrant and J. Kelly and I. Aleiner and V. Smelyanskiy and K. Kechedzhi and Y. Chen and S. Boixo},
  title   = {{Phase transitions in random circuit sampling}},
  journal = {Nature},
  volume  = {634},
  pages   = {328--333},
  year    = {2024},
  doi     = {10.1038/s41586-024-07998-6}
}

@article{Gao25,
  author  = {Dongxin Gao and Daojin Fan and Chen Zha and Jiahao Bei and Guoqing Cai and Jianbin Cai and Sirui Cao and Xiangdong Zeng and Fusheng Chen and Jiang Chen and Kefu Chen and Xiawei Chen and Xiqing Chen and Zhe Chen and Zhiyuan Chen and Zihua Chen and Wenhao Chu and Hui Deng and Zhibin Deng and Pei Ding and Xun Ding and Zhuzhengqi Ding and Shuai Dong and Yupeng Dong and Bo Fan and Yuanhao Fu and Song Gao and Lei Ge and Ming Gong and Jiacheng Gui and Cheng Guo and Shaojun Guo and Xiaoyang Guo and Tan He and Linyin Hong and Yisen Hu and He-Liang Huang and Yong-Heng Huo and Tao Jiang and Zuokai Jiang and Honghong Jin and Yunxiang Leng and Dayu Li and Dongdong Li and Fangyu Li and Jiaqi Li and Jinjin Li and Junyan Li and Junyun Li and Na Li and Shaowei Li and Wei Li and Yuhuai Li and Yuan Li and Futian Liang and Xuelian Liang and Nanxing Liao and Jin Lin and Weiping Lin and Dailin Liu and Hongxiu Liu and Maliang Liu and Xinyu Liu and Xuemeng Liu and Yancheng Liu and Haoxin Lou and Yuwei Ma and Lingxin Meng and Hao Mou and Kailiang Nan and Binghan Nie and Meijuan Nie and Jie Ning and Le Niu and Wenyi Peng and Haoran Qian and Hao Rong and Tao Rong and Huiyan Shen and Qiong Shen and Hong Su and Feifan Su and Chenyin Sun and Liangchao Sun and Tianzuo Sun and Yingxiu Sun and Yimeng Tan and Jun Tan and Longyue Tang and Wenbing Tu and Cai Wan and Jiafei Wang and Biao Wang and Chang Wang and Chen Wang and Chu Wang and Jian Wang and Liangyuan Wang and Rui Wang and Shengtao Wang and Xinzhe Wang and Zuolin Wei and Jiazhou Wei and Dachao Wu and Gang Wu and Jin Wu and Shengjie Wu and Yulin Wu and Shiyong Xie and Lianjie Xin and Yu Xu and Chun Xue and Kai Yan and Weifeng Yang and Xinpeng Yang and Yang Yang and Yangsen Ye and Zhenping Ye and Chong Ying and Jiale Yu and Qinjing Yu and Wenhu Yu and Shaoyu Zhan and Feifei Zhang and Haibin Zhang and Kaili Zhang and Pan Zhang and Wen Zhang and Yiming Zhang and Yongzhuo Zhang and Lixiang Zhang and Guming Zhao and Peng Zhao and Xianhe Zhao and Xintao Zhao and Youwei Zhao and Zhong Zhao and Luyuan Zheng and Fei Zhou and Liang Zhou and Na Zhou and Naibin Zhou and Shifeng Zhou and Shuang Zhou and Zhengxiao Zhou and Chengjun Zhu and Qingling Zhu and Guihong Zou and Haonan Zou and Qiang Zhang and Chao-Yang Lu and Cheng-Zhi Peng and XiaoBo Zhu and Jian-Wei Pan},
  title   = {{Establishing a new benchmark in quantum computational advantage with 105-qubit Zuchongzhi 3.0 processor}},
  journal = {Physical Review Letters},
  volume  = {134},
  pages   = {090601},
  year    = {2025},
  doi     = {10.1103/PhysRevLett.134.090601}
}

@article{Dec25,
  author  = {Matthew DeCross and Reza Haghshenas and Minzhao Liu and Enrico Rinaldi and Johnnie Gray and Yuri Alexeev and Charles H. Baldwin and John P. Bartolotta and Matthew Bohn and Eli Chertkov and Julia Cline and Jonhas Colina and Davide DelVento and Joan M. Dreiling and Cameron Foltz and John P. Gaebler and Thomas M. Gatterman and Christopher N. Gilbreth and Joshua Giles and Dan Gresh and Alex Hall and Aaron Hankin and Azure Hansen and Nathan Hewitt and Ian Hoffman and Craig Holliman and Ross B. Hutson and Trent Jacobs and Jacob Johansen and Patricia J. Lee and Elliot Lehman and Dominic Lucchetti and Danylo Lykov and Ivaylo S. Madjarov and Brian Mathewson and Karl Mayer and Michael Mills and Pradeep Niroula and Juan M. Pino and Conrad Roman and Michael Schecter and Peter E. Siegfried and Bruce G. Tiemann and Curtis Volin and James Walker and Ruslan Shaydulin and Marco Pistoia and Steven A. Moses and David Hayes and Brian Neyenhuis and Russell P. Stutz and Michael Foss-Feig},
  title   = {{Computational power of random quantum circuits in arbitrary geometries}},
  journal = {Physical Review X},
  volume  = {15},
  pages   = {021052},
  year    = {2025},
  doi     = {10.1103/PhysRevX.15.021052}
}

@article{Mad22,
  author  = {Lars Skovgaard Madsen and Fabian Laudenbach and Mohsen Falamarzi Askarani and Fabien Rortais and Trevor Vincent and Jacob F. F. Bulmer and Filippo M. Miatto and Leonhard Neuhaus and Lukas G. Helt and Matthew J. Collins and Adriana E. Lita and Thomas Gerrits and Sae Woo Nam and Varun D. Vaidya and Matteo Menotti and Ish Dhand and Zachary Vernon and Nicolás Quesada and Jonathan Lavoie},
  title   = {{Quantum computational advantage with a programmable photonic processor}},
  journal = {Nature},
  volume  = {606},
  pages   = {75--81},
  year    = {2022},
  doi     = {10.1038/s41586-022-04725-x}
}

@article{AA13,
  author  = {Scott Aaronson and Alex Arkhipov},
  title   = {{The computational complexity of linear optics}},
  journal = {Theory of Computing},
  volume  = {9},
  pages   = {143--252},
  year    = {2013},
  doi     = {10.4086/toc.2013.v009a004}
}

@article{AG20,
  author  = {Scott Aaronson and Sam Gunn},
  title   = {{On the classical hardness of spoofing linear cross-entropy benchmarking}},
  journal = {Theory of Computing},
  volume  = {16},
  pages   = {1--34},
  year    = {2020},
  doi     = {10.4086/toc.2020.v016a011}
}

@article{Nad22,
  title = {{Experimental quantum key distribution certified by Bell’s theorem}},
  volume = {607},
  ISSN = {1476-4687},
  url = {http://dx.doi.org/10.1038/s41586-022-04941-5},
  DOI = {10.1038/s41586-022-04941-5},
  number = {7920},
  journal = {Nature},
  publisher = {Springer Science and Business Media LLC},
  author = {Nadlinger,  D. P. and Drmota,  P. and Nichol,  B. C. and Araneda,  G. and Main,  D. and Srinivas,  R. and Lucas,  D. M. and Ballance,  C. J. and Ivanov,  K. and Tan,  E. Y.-Z. and Sekatski,  P. and Urbanke,  R. L. and Renner,  R. and Sangouard,  N. and Bancal,  J.-D.},
  year = {2022},
  pages = {682–686}
}

@article{Zha22,
  title = {{A device-independent quantum key distribution system for distant users}},
  volume = {607},
  ISSN = {1476-4687},
  url = {http://dx.doi.org/10.1038/s41586-022-04891-y},
  DOI = {10.1038/s41586-022-04891-y},
  number = {7920},
  journal = {Nature},
  publisher = {Springer Science and Business Media LLC},
  author = {Zhang,  Wei and van Leent,  Tim and Redeker,  Kai and Garthoff,  Robert and Schwonnek,  René and Fertig,  Florian and Eppelt,  Sebastian and Rosenfeld,  Wenjamin and Scarani,  Valerio and Lim,  Charles C.-W. and Weinfurter,  Harald},
  year = {2022},
  pages = {687–691}
}

@article{CSUU08,
  title = {{Perfect Parallel Repetition Theorem for Quantum Xor Proof Systems}},
  volume = {17},
  ISSN = {1420-8954},
  url = {http://dx.doi.org/10.1007/s00037-008-0250-4},
  DOI = {10.1007/s00037-008-0250-4},
  number = {2},
  journal = {Computational Complexity},
  author = {Cleve,  Richard and Slofstra,  William and Unger,  Falk and Upadhyay,  Sarvagya},
  year = {2008},
  pages = {282–299}
}

\end{document}